\pdfoutput=1
\documentclass[11pt,fleqn]{article}

\usepackage[utf8]{inputenc}
\usepackage[T1]{fontenc}

\usepackage{comment}
\usepackage[dvipsnames]{xcolor}

\usepackage{amsmath,amssymb,amsthm,mathtools,thm-restate}
\usepackage{fullpage}
\usepackage[indent,parfill=0pt]{parskip}  
\usepackage{indentfirst}

\usepackage{fouriernc}
\usepackage[cal=pxtx,scr=rsfs,bb=txof,frak=pxtx]{mathalpha}
\let\newmathbb\mathbb
\AtBeginDocument{%
    \let\mathbb\relax
    \newcommand{\mathbb}[1]{\bm{\newmathbb{#1}}}
}

\usepackage[style=alphabetic,natbib=true,sortcites=true,backref=true,maxnames=99,useprefix=true,maxalphanames=99]{biblatex}

\DefineBibliographyStrings{english}{%
  backrefpage = {$\Lsh$~p.},
  backrefpages = {$\Lsh$~pp.},
}
\usepackage[colorlinks=true,citecolor=ForestGreen,linkcolor=Blue,urlcolor=Magenta,linktoc=page,bookmarksnumbered=true]{hyperref}
\usepackage[nameinlink]{cleveref}

\usepackage[basic,bold,disableredefinitions]{complexity}
\usepackage{booktabs,enumitem}
\usepackage{ascmac}

\usepackage{url}

\usepackage[colorinlistoftodos,prependcaption]{todonotes}

\usepackage{xspace}
\usepackage{bm}
\usepackage{mleftright}
\mleftright

\crefname{theorem}{Theorem}{Theorems}
\crefname{proposition}{Proposition}{Propositions}
\crefname{lemma}{Lemma}{Lemmas}
\crefname{claim}{Claim}{Claims}
\crefname{corollary}{Corollary}{Corollaries}
\crefname{observation}{Observation}{Observations}
\crefname{remark}{Remark}{Remarks}
\crefname{example}{Example}{Examples}
\crefname{hypothesis}{Hypothesis}{Hypotheses}
\crefname{definition}{Definition}{Definitions}
\crefname{problem}{Problem}{Problems}

\crefname{section}{Section}{Sections}
\crefname{appendix}{Appendix}{Appendices}
\crefname{equation}{Eq.}{Eqs.}
\crefname{table}{Table}{Tables}
\crefname{figure}{Figure}{Figures}

\renewcommand{\geq}{\geqslant}
\renewcommand{\leq}{\leqslant}
\renewcommand{\phi}{\varphi}
\renewcommand{\epsilon}{\varepsilon}

\renewcommand{\tilde}{\widetilde}

\newcommand{\prb}[1]{\textsf{#1}\xspace}

\newcommand{\reco}{\leftrightsquigarrow}

\DeclareMathOperator{\bigO}{\mathrm{O}}

\DeclareMathOperator{\opt}{\mathsf{opt}}
\DeclareMathOperator{\proj}{\Pi}

\let\polylog\relax\DeclareMathOperator{\polylog}{\mathrm{polylog}}

\newcommand{\ttt}{t}
\newcommand{\TTT}{T}

\usepackage[ruled]{algorithm}
\usepackage[commentColor=Blue]{algpseudocodex}
\algnewcommand{\algorithmicand}{\textbf{ and }}
\algnewcommand{\algorithmicor}{\textbf{ or }}
\algnewcommand{\algorithmicto}{\textbf{ to }}

\algrenewcommand\textproc{\textsl}

\newcommand{\calA}{\mathcal{A}}

\newcommand{\bbN}{\mathbb{N}}

\newcommand{\scrR}{\mathscr{R}}
\newcommand{\scrS}{\mathscr{S}}

\usepackage{tikz}
\usetikzlibrary{math}
\usetikzlibrary{positioning}
\usetikzlibrary{decorations.markings}
\usetikzlibrary{shapes.geometric}
\usetikzlibrary{arrows.meta}
\usetikzlibrary{arrows,automata,decorations.pathmorphing,fit,positioning,arrows.meta,calc}

\newcommand{\convexpath}[2]{
[   
    create hullnodes/.code={
        \global\edef\namelist{#1}
        \foreach [count=\counter] \nodename in \namelist {
            \global\edef\numberofnodes{\counter}
            \node at (\nodename) [draw=none,name=hullnode\counter] {};
        }
        \node at (hullnode\numberofnodes) [name=hullnode0,draw=none] {};
        \pgfmathtruncatemacro\lastnumber{\numberofnodes+1}
        \node at (hullnode1) [name=hullnode\lastnumber,draw=none] {};
    },
    create hullnodes
]
($(hullnode1)!#2!-90:(hullnode0)$)
\foreach [
    evaluate=\currentnode as \previousnode using \currentnode-1,
    evaluate=\currentnode as \nextnode using \currentnode+1
    ] \currentnode in {1,...,\numberofnodes} {
-- ($(hullnode\currentnode)!#2!-90:(hullnode\previousnode)$)
  let \p1 = ($(hullnode\currentnode)!#2!-90:(hullnode\previousnode) - (hullnode\currentnode)$),
    \n1 = {atan2(\y1,\x1)},
    \p2 = ($(hullnode\currentnode)!#2!90:(hullnode\nextnode) - (hullnode\currentnode)$),
    \n2 = {atan2(\y2,\x2)},
    \n{delta} = {-Mod(\n1-\n2,360)}
  in 
    {arc [start angle=\n1, delta angle=\n{delta}, radius=#2]}
}
-- cycle
}

\newtheorem{theorem}{Theorem}[section]
\newtheorem{proposition}[theorem]{Proposition}
\newtheorem{lemma}[theorem]{Lemma}
\newtheorem{claim}[theorem]{Claim}

\newtheorem{observation}[theorem]{Observation}

\theoremstyle{definition}
\newtheorem{definition}[theorem]{Definition}

\numberwithin{equation}{section}

\title{Tight Inapproximability of Target Set Reconfiguration}
\author{
Naoto Ohsaka \\
\small{CyberAgent, Inc., Japan} \\
\small{\href{mailto:naoto.ohsaka@gmail.com}{\texttt{ohsaka\_naoto@cyberagent.co.jp}}}
}
\date{}

\begin{document}
\maketitle
\begin{abstract}Given a graph $G$ with a vertex threshold function $\tau$,
consider a dynamic process in which any inactive vertex $v$ becomes activated whenever
at least $\tau(v)$ of its neighbors have been activated.
A vertex set $S$ is called a \emph{target set} if
all vertices of $G$ would eventually be activated when initially activating exactly the vertices of $S$.
In the \prb{Minmax Target Set Reconfiguration} problem,
for a graph $G$ and a pair of its target sets $X$ and $Y$,
we wish to transform $X$ into $Y$ by repeatedly adding or removing a single vertex,
using only target sets of $G$,
so as to minimize the \emph{maximum size} of any intermediate target set.
We prove that 
it is \NP-hard to approximate
\prb{Minmax Target Set Reconfiguration}
within a factor of $2-o\left(\frac{1}{\operatorname{polylog} n}\right)$,
where $n$ is the number of vertices.
Our result establishes a tight lower bound on approximability of \prb{Minmax Target Set Reconfiguration},
which admits a simple $2$-factor approximation algorithm.
The proof is based on a gap-preserving reduction from
\prb{Target Set Selection} to
\prb{Minmax Target Set Reconfiguration},
where \NP-hardness of approximation for the former problem is proven by
{Chen} (SIDMA 2009)~\cite{chen2009approximability} and
{Charikar, Naamad, and Wirth} (APPROX/RANDOM 2016)~\cite{charikar2016approximating}.
\end{abstract}
\section{Introduction}

\paragraph{Background.}
Given a graph $G = (V,E,\tau)$, where $\tau \colon V \to \bbN$ is a vertex threshold function,\footnote{
Throughout this paper, we assume that a vertex threshold function is associated with every graph.
}
consider a dynamic process in which
any inactive vertex $v$ becomes activated whenever at least $\tau(v)$ of its neighbors have been activated.
Formally, the \emph{activation process} over $G$, which proceeds in discrete-time steps, is defined as follows.
Each vertex takes either of two states, \emph{active} or \emph{inactive}.
For a vertex set $S \subseteq V$,
let $\calA^{(i)}(S)$ denote the set of already activated vertices at step $i$.
Initially, the vertices of $S$ are activated and the others remain inactive;
namely, $\calA^{(0)}(S) \triangleq S$.
At step $i-1$,
we examine whether each inactive vertex $v$ has at least $\tau(v)$ active neighbors.
If this is the case, then it will be active at step $i$.
Thus, $\calA^{(i)}(S)$ for each $i \geq 1$ is recursively defined as follows:
\begin{align}
    \calA^{(i)}(S) \triangleq
    \calA^{(i-1)}(S) \cup \Bigl\{ v \in V \Bigm| \bigl|N(v) \cap \calA^{(i-1)}(S)\bigr| \geq \tau(v) \Bigr\},
\end{align}
where $N(v)$ is the set of $v$'s neighbors.
This process is \emph{irreversible}; i.e.,
an active vertex may not become inactive.
We define the \emph{active vertex set} of $S$ as
$\calA(S) \triangleq \lim_{i \to \infty} \calA^{(i)}(S)$, 
which is in fact equal to $\calA^{(|V|)}(S)$.
We say that $S$ \emph{activates} a vertex $v$ if $v \in \calA(S)$.
A \emph{target set} for $G$ is defined as a vertex set $S \subseteq V$ such that
all vertices of $G$ would eventually be activated when initially activating exactly the vertices of $S$;
namely, $\calA(S) = V$.

The \prb{Target Set Selection} problem \cite{kempe2003maximizing,kempe2015maximizing,chen2009approximability} asks to identify the minimum target set for $G$.

\begin{itembox}[l]{\prb{Target Set Selection}}
\begin{tabular}{ll}
    \textbf{Input:}
    & a graph $G = (V,E,\tau)$.
    \\
    \textbf{Output:}
    & a minimum target set for $G$.
\end{tabular}
\end{itembox}
We will use $\opt(G)$ to denote the minimum size of any target set for $G$.
Since the activation process models the spread of influence, information, and opinion over a network,
\prb{Target Set Selection} finds applications in social network analysis 
\cite{chen2009approximability,kempe2015maximizing,kempe2003maximizing} and
distributed computing \cite{peleg2002local,peleg1998size}, and
is known by various names such as
\emph{irreversible $k$-conversion sets} \cite{dreyer2009irreversible,centeno2011irreversible} and
\emph{dynamic monopolies} \cite{peleg1998size,peleg2002local}.
This problem generalizes
\prb{Minimum Vertex Cover} and \prb{Minimum Feedback Vertex Set} in that
a target set is a vertex cover (resp.~a feedback vertex set)
if $\tau(v)$ is equal to the degree of $v$ (resp.~the degree of $v$ minus $1$) for each vertex $v$ \cite{dreyer2000applications,dreyer2009irreversible}.

On the hardness side,
\prb{Target Set Selection} is known to be \NP-hard \cite{peleg2002local};
indeed, it is \NP-hard to approximate the minimum target set within a factor
of $2^{\log^{1-\epsilon} n}$ \cite{chen2009approximability} and
a factor of $n^{\frac{1}{2}-\epsilon}$ under the planted dense subgraph conjecture \cite{charikar2016approximating} for any $\epsilon > 0$,
where $n$ is the number of vertices.
On the algorithmic side,
efficient algorithms have been developed in
the parameterized (approximation) regime
\cite{ben-zwi2011treewidth,chu2023fpt,dvorak2022target,nichterlein2013tractable,bazgan2014parameterized,suzuki2025parameterized},
the restricted graph classes
\cite{chiang2013target,chiang2013some,dvorak2024complexity,feige2021target,centeno2011irreversible,penso2014p3,kyncl2017irreversible,ueno1988nonseparating,takaoka2015note,nichterlein2013tractable},
the bounded threshold case \cite{bliznets2023solving}, and
dynamic environments \cite{deligkas2024being,schierreich2023maximizing,ohsaka2016maximizing}.

We study a reconfiguration analogue of \prb{Target Set Selection}.
In \emph{reconfiguration problems} \cite{ito2011complexity},
we would like to determine the connectivity between feasible solutions for a combinatorial problem.
For a graph $G = (V,E,\tau)$ and a pair of its target sets $X$ and $Y$ of size $k$,
a \emph{reconfiguration sequence from $X$ to $Y$}
is defined as any sequence,
denoted by $\scrS = ( S^{(1)}, \ldots, S^{(\TTT)} ) $,
such that
$S^{(1)} = X$,
$S^{(\TTT)} = Y$, and
$S^{(\ttt)}$ for each $2 \leq \ttt \leq \TTT$ is a target set
obtained from $S^{(\ttt-1)}$ by adding or removing a single vertex
(i.e., $|S^{(\ttt-1)} \triangle S^{(\ttt)}| = 1$).\footnote{
Such a model of reconfiguration is called \emph{token addition and removal} \cite{ito2011complexity}.
}
The \prb{Target Set Reconfiguration} problem \cite{ohsaka2023reconfigurability}
requests to decide if
there exists a reconfiguration sequence from $X$ to $Y$
consisting only of target sets of size at most $k+1$.
This problem is shown to be \PSPACE-complete, as
it includes a \PSPACE-complete \prb{Vertex Cover Reconfiguration} problem 
\cite{hearn2005pspace,ito2011complexity,kaminski2012complexity}
as a special case.
Our prior work \cite{ohsaka2023reconfigurability}
investigates 
the dividing line between polynomial-time solvability and \PSPACE-completeness
of \prb{Target Set Reconfiguration}
in the restricted case; e.g.,
it is solvable in polynomial time on trees, whereas
\PSPACE-complete even on bipartite planar graphs and split graphs.
We refer the readers to the surveys  \cite{nishimura2018introduction,heuvel2013complexity,mynhardt2019reconfiguration,bousquet2024survey}
as well as
the Combinatorial Reconfiguration wiki \cite{hoang2024combinatorial}
for algorithmic and hardness results of reconfiguration problems.

Our focus in this study is on \emph{approximate reconfigurability} 
\cite{ohsaka2023gap,ohsaka2024gap,ohsaka2024alphabet,ohsaka2025approximate,ohsaka2025yet}
of \prb{Target Set Reconfiguration},
which affords to use any large target set but requires to minimize the \emph{maximum size} of any intermediate target set.
For any reconfiguration sequence $\scrS = ( S^{(1)}, \ldots, S^{(\TTT)} )$,
its \emph{size} $\|\scrS\|$ is defined as
the maximum size of $S^{(\ttt)}$ in $\scrS$; namely,
\begin{align}
    \|\scrS\| \triangleq \max_{1 \leq \ttt \leq \TTT} |S^{(\ttt)}|.
\end{align}
In \prb{Minmax Target Set Reconfiguration} --- a canonical approximate version of \prb{Target Set Reconfiguration} --- we wish to find a reconfiguration sequence $\scrS$ 
from $X$ to $Y$
such that $\|\scrS\|$ is minimized.

\begin{itembox}[l]{\prb{Minmax Target Set Selection}}
\begin{tabular}{ll}
    \textbf{Input:}
    & a graph $G = (V,E,\tau)$ and
    a pair of its target sets $X$ and $Y$ of size $k$.
    \\
    \textbf{Output:}
    & a reconfiguration sequence $\scrS$ from $X$ to $Y$
    such that $\|\scrS\|$ is minimized.
\end{tabular}
\end{itembox}

\noindent
We will use $\opt_G(X \reco Y)$ to denote
the minimum value of $\|\scrS\|$ over
all possible reconfiguration sequences $\scrS$ from $X$ to $Y$; namely,
\begin{align}
    \opt_G(X \reco Y) \triangleq \min_{\scrS = ( X, \ldots, Y )} \|\scrS\|.
\end{align}
Obviously, we can build
a $2$-approximation reconfiguration sequence $\scrS$ such that 
$\|\scrS\| \leq |X \cup Y| \leq 2k$
in polynomial time by the following simple procedure, e.g., \cite[Theorem~6]{ito2011complexity}.
\begin{itembox}[l]{$2$-approximation reconfiguration sequence for \prb{Minmax Target Set Reconfiguration}}
\begin{algorithmic}[1]
    \LComment{start with $X$.}
    \State add vertices of $Y \setminus X$ one by one.
    \LComment{obtain $X \cup Y$.}
    \State remove vertices of $X \setminus Y$ one by one.
    \LComment{end with $Y$.}
\end{algorithmic}
\end{itembox}

\noindent
Of particular interest is thus whether $(2-\epsilon)$-approximation is possible (in polynomial time).

\paragraph{Our Result.}
We prove that 
\prb{Minmax Target Set Reconfiguration} cannot be approximated within a factor better than $2$
unless \cP~$=$~\NP,
formally stated below.

\begin{theorem}\label{thm:main}
    \prb{Minmax Target Set Reconfiguration}
    is \NP-hard to approximate within a factor of $2-o\left(\frac{1}{\polylog n}\right)$,
    where $n$ is the number of vertices.
\end{theorem}\noindent
Since \prb{Minmax Target Set Reconfiguration} admits a $2$-factor approximation algorithm as shown above,
\cref{thm:main} establishes a tight lower bound on approximability of \prb{Minmax Target Set Reconfiguration}.
The proof of \cref{thm:main} is based on a gap-preserving reduction from
\prb{Target Set Selection} to
\prb{Minmax Target Set Reconfiguration},
where \NP-hardness of approximation for the former problem is proven 
by \citet{chen2009approximability,charikar2016approximating}.

\paragraph{Additional Related Work.}
Other reconfiguration problems whose inapproximability was studied include
\prb{$k$-SAT Reconfiguration} \cite{ito2011complexity,hirahara2025asymptoticallya},
\prb{$k$-Coloring Reconfiguration} \cite{hirahara2025asymptotically},
\prb{Clique Reconfiguration} \cite{ito2011complexity,hirahara2024probabilistically,hoang2026inapproximability},
\prb{2-CSP Reconfiguration} \cite{ohsaka2024gap,karthikc.s.2023inapproximability,hirahara2026optimal},
\prb{Set Cover Reconfiguration} \cite{ohsaka2024gap,karthikc.s.2023inapproximability,hirahara2024optimal}, and
\prb{Submodular Reconfiguration} \cite{ohsaka2022reconfiguration}.
Among them,
\citet{karthikc.s.2023inapproximability} proved that
\prb{2-CSP Reconfiguration} is \NP-hard to approximate within a factor of
$\frac{1}{2}+\epsilon$ for any $\epsilon \in (0,1)$, which is tight as
it is $\left(\frac{1}{2}-\epsilon\right)$-factor approximable \cite{karthikc.s.2023inapproximability}.
Note that these \NP-hardness results rely on
a gap-preserving reduction from an \NP-hard \emph{source problem} to its reconfiguration analogue.
Our study follows the same idea,
deriving a tight inapproximability result for \prb{Minmax Target Set Reconfiguration}.

\section{Proof of \texorpdfstring{\cref{thm:main}}{Theorem 1.1}}
Our proof of \cref{thm:main} reduces \prb{Target Set Selection} to
\prb{Minmax Target Set Reconfiguration}.
Specifically, we derive the following gap-preserving reducibility.

\begin{proposition}
\label{prp:reduct}
    For any positive integers $k_c,k_s,\ell \in \bbN$ such that $\ell > k_s \geq k_c$,
    there is a polynomial-time algorithm that 
    takes an instance $G$ of \prb{Target Set Selection} and
    produces an instance $(H;X,Y)$ of \prb{Minmax Target Set Reconfiguration},
    where $X$ and $Y$ are target sets of size $\ell$,
    such that the following hold\textup{:}
    \begin{itemize}
        \item \textup{(}Completeness\textup{)}
        If $\opt(G) \leq k_c$, then $\opt_H(X \reco Y) \leq k_c+\ell$.
        \item \textup{(}Soundness\textup{)}
        If $\opt(G) > k_s$, then $\opt_H(X \reco Y) > k_s+\ell$.
    \end{itemize}
    Moreover, $H$ contains $\bigO(|V(G)|^2 \cdot \ell)$ vertices.
\end{proposition}

\noindent
\cref{thm:main} follows from \cref{prp:reduct} and
\NP-hardness of approximation for \prb{Target Set Selection} 
\cite{dinur2004hardness,chen2009approximability,charikar2016approximating}.

\begin{proof}[Proof of \cref{thm:main}]
We first recapitulate \NP-hardness of approximation for \prb{Target Set Selection}.
By \citet[Theorem~2]{dinur2004hardness},
it is \NP-hard to distinguish for a \prb{Label Cover} instance on $N$ vertices,
whether its optimal solution has size
at most $N$ or more than $g(N) \cdot N$,
where $g(N) \triangleq 2^{\log^{0.99} N}$.
Note that $g$ grows faster than any polylogarithmic function;
namely, $g(N) = \omega(\polylog N)$.
By using \citet[Theorem~2.1]{chen2009approximability},
we construct a \prb{Target Set Selection} instance $G$
from a \prb{Label Cover} instance on $N$ vertices,
where $G$ contains $\bigO(N^2)$ vertices.\footnote{
Although \citet{chen2009approximability} assumes in Theorem~2.1 that
$\NP \not\subseteq \DTIME(n^{\polylog(n)})$,
as pointed out by \citet{charikar2016approximating},
we can remove this assumption by reducing from the \prb{Label Cover} problem \cite[Theorem~2]{dinur2004hardness}.
}
Since \citeauthor{chen2009approximability}'s reduction \cite[Theorem~2.1]{chen2009approximability}
preserves the size of the optimal solution of \prb{Label Cover}
by (at most) a factor of $2$,
it is \NP-hard to distinguish whether
$\opt(G) \leq N$ (i.e., completeness) or 
$\opt(G) > \left\lceil\frac{1}{2}g(N) \cdot N \right\rceil$ (i.e., soundness).
Applying \cref{prp:reduct} to $G$ with
$k_c \triangleq N$,
$k_s \triangleq \left\lceil\frac{1}{2}g(N) \cdot N\right\rceil$, and
$\ell \triangleq k_s+1 = \left\lceil \frac{1}{2} g(N)\cdot N +1 \right\rceil$,
we obtain a \prb{Minmax Target Set Reconfiguration} instance $(H;X,Y)$,
where $H$ contains $\bigO(|V(G)|^2\cdot\ell) = \bigO(N^6)$ vertices and
$X$ and $Y$ are a pair of its target sets of size $\ell$,
such that the following hold:
\begin{itemize}
    \item (Completeness)
        If $\opt(G) \leq N$,
        then $\opt_H(X \reco Y) \leq N + \ell$;
    \item (Soundness)
        If $\opt(G) > \left\lceil \frac{1}{2}g(N) \cdot N \right\rceil$,
        then $\opt_H(X \reco Y) > \left\lceil\frac{1}{2}g(N) \cdot N\right\rceil + \ell$.
\end{itemize}
Consequently,
the inapproximability factor (i.e., the ratio of soundness to completeness) 
can be bounded as
\begin{align}
\begin{aligned}
    \frac{\left\lceil\frac{1}{2}g(N) \cdot N\right\rceil + \ell}{N+\ell}
    & = \frac{\left\lceil\frac{1}{2}g(N) \cdot N\right\rceil + \left\lceil \frac{1}{2} g(N)\cdot N +1 \right\rceil}{N + \left\lceil \frac{1}{2} g(N)\cdot N +1 \right\rceil} \\ 
    & \geq \frac{g(N) \cdot N}{\frac{1}{2} g(N)\cdot N + N + 2} \\
    & = 2 - \frac{2N+4}{\frac{1}{2} g(N)\cdot N + N + 2} \\
    & = 2 - \bigO\left(\frac{1}{g(N)}\right)
    = 2-o\left(\frac{1}{\polylog N}\right).
\end{aligned}
\end{align}
Observing that $N = \Omega\left(|V(H)|^{\frac{1}{6}}\right)$ completes the proof.
\end{proof}

\subsection{Reduction}
Our gap-preserving reduction from \prb{Target Set Selection} to \prb{Minmax Target Set Reconfiguration}
is described as follows.
We first introduce \emph{one-way gadgets}
\cite{kyncl2017irreversible,charikar2016approximating,bazgan2014parameterizeda}.

\begin{definition}[One-way gadget \cite{kyncl2017irreversible,charikar2016approximating,bazgan2014parameterizeda}]
    A \emph{one-way gadget} is defined as a graph $D=(V,E,\tau)$ such that
    \begin{align}
    \begin{aligned}
        V & \triangleq \bigl\{ t,h,b_1,b_2 \bigr\}, \\
        E & \triangleq \bigl\{ (t,b_1), (t,b_2), (h,b_1), (h,b_2) \bigr\}, \\
        \tau(t)& =\tau(b_1)=\tau(b_2)=1 \text{ and } \tau(h)=2.
    \end{aligned}
    \end{align}
    The vertices of $V$ are called the \emph{internal vertices} of $D$.
    We say that $D$ \emph{connects from vertex $v$ to vertex $w$} if there exist two edges $(v,t)$ and $(w,h)$;
    $v$ is called the \emph{tail} and $w$ is called the \emph{head} of $D$.
\end{definition}

Observe that an active tail $v$ activates all internal vertices $t$, $h$, $b_1$, and $b_2$, but an active head $w$ does not.

Let $G=(V,E,\tau)$ be a graph on $n$ vertices
representing an instance of \prb{Target Set Selection}.
Without loss of generality, we can assume that $G$ has no isolated vertices.
Hereafter, let $d_v$ denote the degree of $v$.
Given a positive integer $\ell \in \bbN$,
we construct a new graph $H$ as follows.

\begin{itembox}[l]{Construction of $H$}
\small
\begin{algorithmic}[1]
    \Require a graph $G=(V,E,\tau)$ on $n$ vertices and a positive integer $\ell$.
    \Ensure a graph $H$ with a vertex threshold function $\tau'$.
    \State create a copy of $G$.  \Comment{the vertex threshold function $\tau$ is also copied.}
    \State create four different groups of new vertices, denoted by $X$, $Y$, $A$, and $B$, as follows:
    \begin{align}
        X & \triangleq \bigl\{x_i \bigm| 1 \leq i \leq \ell\bigr\}
            & \text{s.t. } \forall x_i \in X,\; \tau'(x_i) & \triangleq |V| = n, \\
        Y & \triangleq \bigl\{y_i \bigm| 1 \leq i \leq \ell\bigr\}
            & \text{s.t. } \forall y_i \in Y,\; \tau'(y_i) & \triangleq |V| = n, \\
        A & \triangleq \bigl\{a_{v,j} \bigm| v \in V \text{ and } 1 \leq j \leq d_v\bigr\}
            & \text{s.t. } \forall a_{v,j} \in A,\; \tau'(a_{v,j}) & \triangleq |X| = \ell, \\
        B & \triangleq \bigl\{b_{v,j} \bigm| v \in V \text{ and } 1 \leq j \leq d_v \bigr\}
            & \text{s.t. } \forall b_{v,j} \in B,\; \tau'(b_{v,j}) & \triangleq |Y| = \ell.
    \end{align}
    \LComment{connect some vertices of $V\uplus X \uplus Y \uplus A \uplus B$ by one-way gadgets.}
    \For{\textbf{each} $v \in V$ and $x_i \in X$}
        \State create a one-way gadget $D_{v,x_i}$ connecting from $v$ to $x_i$.
    \EndFor
    \For{\textbf{each} $x_i \in X$ and $a_{v,j} \in A$}
        \State create a one-way gadget $D_{x_i,a_{v,j}}$ connecting from $x_i$ to $a_{v,j}$.
    \EndFor
    \For{\textbf{each} $a_{v,j} \in A$}
        \State create a one-way gadget $D_{a_{v,j},v}$ connecting from $a_{v,j}$ to $v$.
    \EndFor
    \For{\textbf{each} $v \in V$ and $y_i \in Y$}
        \State create a one-way gadget $D_{v,y_i}$ connecting from $v$ to $y_i$.
    \EndFor
    \For{\textbf{each} $y_i \in Y$ and $b_{v,j} \in B$}
        \State create a one-way gadget $D_{y_i,b_{v,j}}$ connecting from $y_i$ to $b_{v,j}$.
    \EndFor
    \For{\textbf{each} $b_{v,j} \in B$}
        \State create a one-way gadget $D_{b_{v,j},v}$ connecting from $b_{v,j}$ to $v$.
    \EndFor
    \State \textbf{return} the resulting graph $H$.
\end{algorithmic}
\end{itembox}

See \cref{fig:H} for an illustration of the construction of $H$.
Note that
$|X| = |Y| = \ell$ and
$|A| = |B| = \sum_{v \in V} d_v = 2|E|$; in particular,
\begin{align}
\begin{aligned}
    |V(H)|
    & = |V(G)|+|X|+|Y|+|A|+|B| \\
    & \quad+ \underbrace{4 \cdot (|V(G) \times X| + |X \times A| + |A| + |V(G) \times Y| + |Y \times B| + |B|)}_{\text{the number of vertices in the one-way gadgets}} \\
    & = \bigO(|V(G)|^2 \cdot \ell).
\end{aligned}
\end{align}
Then, an instance of \prb{Minmax Target Set Reconfiguration} is defined as $(H;X,Y)$,
where each $X$ and $Y$ is a target set for $H$ as proven below,
completing the description of the reduction.

\begin{observation}
\label{obs:target}
    Each of $X$, $Y$, and any target set for $G$ is a target set for $H$.
\end{observation}
\begin{proof}
We first show that $X$ is a target set for $H$.
Initially activating the vertices of $X$ would make (the internal vertices of) all one-way gadgets
connecting from $X$ to $A$ active.
Since all vertices of $A$ have threshold $\ell = |X|$, they would also become activated.
For each vertex $v$ of $G$, there are $d_v$ vertices of $A$ connected to $v$ via one-way gadgets;
thus, all vertices of $G$ would get activated, which results in
activation of the other vertices (including $Y\uplus B$);
i.e., $X$ is a target set for $H$.
Similarly, $Y$ is a target set for $H$. 
Observe finally that activating all vertices of $G$ would make
the vertices of $X$ (and thus all vertices of $H$) active,
implying that any target set for $G$ is also a target set for $H$.
This completes the proof.
\end{proof}

\begin{figure}[t]
    \centering
    \resizebox{\textwidth}{!}{%
        \begin{tikzpicture}
\newcommand{\nicered}{Red}
\newcommand{\niceblue}{Blue}

\tikzset{node/.style={circle, very thick, draw=black, fill=white, font=\Huge, text centered, inner sep=0, outer sep=0, minimum size=12mm}};
\tikzset{edge/.style={decoration={
        markings,
        mark=at position 0.68 with {\arrow{Stealth[round, width=4mm, length=4mm]}}  
    },
    postaction={decorate},
    ultra thick
}};


\node[node] (v1) at (-2,0) {$v_1$};
\node[node] (v2) at (+0,1) {$v_2$};
\node[node] (v3) at (+2,0) {$v_3$};
\node[node] (x1) at (-10,2) {$x_1$};
\node[node] (x2) at (-8,2) {$x_2$};
\node[node] (x3) at (-6,2) {$x_3$};
\node[node] (y1) at (+6,2) {$y_1$};
\node[node] (y2) at (+8,2) {$y_2$};
\node[node] (y3) at (+10,2) {$y_3$};
\node[node] (a11) at (-11,-2) {\Large $a_{v_1,1}$};
\node[node] (a21) at (-9,-2) {\Large $a_{v_2,1}$};
\node[node] (a22) at (-7,-2) {\Large $a_{v_2,2}$};
\node[node] (a31) at (-5,-2) {\Large $a_{v_3,1}$};
\node[node] (b11) at (+5,-2) {\Large $b_{v_1,1}$};
\node[node] (b21) at (+7,-2) {\Large $b_{v_2,1}$};
\node[node] (b22) at (+9,-2) {\Large $b_{v_2,2}$};
\node[node] (b31) at (+11,-2) {\Large $b_{v_3,1}$};

\draw[ultra thick] (v1)--(v2)--(v3);
\foreach \P/\angP in {v2/40,v3/45,v1/35}
    \foreach \Q/\angQ in {x1/95,x2/100,x3/105}
        \draw[edge, draw=\niceblue] (\P) to [out=\angQ, in=\angP] (\Q);
\foreach \P/\angP in {v2/130,v3/135,v1/125}
    \foreach \Q/\angQ in {y3/85,y2/80,y1/75}
        \draw[edge, draw=\nicered] (\P) to [out=\angQ, in=\angP] (\Q);

\foreach \P in {x1,x2,x3}
    \foreach \Q in {a11,a21,a22,a31}
        \draw[edge, draw=\niceblue] (\P) to (\Q);
\foreach \P in {y1,y2,y3}
    \foreach \Q in {b11,b21,b22,b31}
        \draw[edge, draw=\nicered] (\P) to (\Q);

\foreach \P/\Q/\ang in {a31/v3/-95, a21/v2/-95, a22/v2/-100, a11/v1/-95}
    \draw[edge, draw=\niceblue] (\P) to [out=-60, in=\ang] (\Q);
\foreach \P/\Q/\ang in {b11/v1/-85, b22/v2/-85, b21/v2/-80, b31/v3/-85}
    \draw[edge, draw=\nicered] (\P) to [out=-120, in=\ang] (\Q);


\draw[densely dashed, thick] \convexpath{x1,x2,x3}{9mm};
\draw[densely dashed, thick] \convexpath{y1,y2,y3}{9mm};
\draw[densely dashed, thick] \convexpath{a11,a21,a22,a31}{9mm};
\draw[densely dashed, thick] \convexpath{b11,b21,b22,b31}{9mm};
\draw[densely dashed, thick] \convexpath{v1,v2,v3}{9mm};

\node[centered] at (0,-1.4) {\Huge $G$};
\node[left=8mm of x1, centered]{\Huge $X$};
\node[right=8mm of y3, centered]{\Huge $Y$};
\node[left=8mm of a11, centered]{\Huge $A$};
\node[right=8mm of b31, centered]{\Huge $B$};

\pgfresetboundingbox
\path[use as bounding box] (-12.91,-4.0) rectangle (12.91,4.7);

\end{tikzpicture}
    }
    \caption{
        Construction of $H$ when $n=3$ and $\ell = 3$.
        One-way gadgets are denoted by 
        \scalebox{0.6}{\protect\tikz \protect\draw[-{Stealth[round, width=4mm, length=4mm]}, ultra thick] (0,0)--(.5,0);}.
    }
    \label{fig:H}
\end{figure}

We first prove the completeness part of \cref{prp:reduct}.
\begin{lemma}
\label{lem:completeness}
    If $\opt(G) \leq k_c$, then $\opt_H(X \reco Y) \leq k_c + \ell$.
\end{lemma}
\begin{proof}
Let $S \subseteq V$ be a target set of size at most $ k_c$ for $G$.
Consider a reconfiguration sequence $\scrS$ from $X$ to $Y$
obtained by the following procedure.

\begin{itembox}[l]{Reconfiguration sequence from $X$ to $Y$}
\begin{algorithmic}[1]
    \LComment{start with $X$.}
    \State add the vertices of $S$ one by one.
    \LComment{obtain $X \uplus S$.}
    \State remove the vertices of $X$ one by one.
    \LComment{obtain $S$.}
    \State add the vertices of $Y$ one by one.
    \LComment{obtain $Y \uplus S$.}
    \State remove the vertices of $S$ one by one.
    \LComment{end with $Y$.}
\end{algorithmic}
\end{itembox}
Each vertex set of $\scrS$ contains either $X$, $Y$, or $S$,
which is thus a target set for $H$ by \cref{obs:target};
i.e., $\scrS$ is a reconfiguration sequence from $X$ to $Y$.
Clearly, it holds that
\begin{align}
    \opt_H(X \reco Y)
    \leq \|\scrS\|
    = \max\bigl\{|X \uplus S|, |Y \uplus S|\bigr\}
    \leq k_c+\ell,
\end{align}
completing the proof.
\end{proof}

\subsection{Proof of Soundness}
In the remainder of this section, we prove the soundness part of \cref{prp:reduct}.
\begin{lemma}
\label{lem:soundness}
    If $\opt(G) > k_s$, then $\opt_H(X \reco Y) > k_s+\ell$.
\end{lemma}
Suppose (for contraposition) that we are given a reconfiguration sequence
$\scrS = ( S^{(1)}, \ldots, S^{(\TTT)} )$ from $X$ to $Y$ such that
$\|\scrS\| = \opt_H(X \reco Y)$.
If $\|\scrS\|$ were relatively small,
then we would like to extract from $\scrS$ a relatively small target set for $G$.
To this end,
we shall render $\scrS$ ``easy-to-handle'' without affecting its size
through a couple of transformations.
We first remove from $\scrS$ all internal vertices of the one-way gadgets.

\begin{claim}
\label{clm:VXYAB}
    There exists a reconfiguration sequence $\scrS$ from $X$ to $Y$
    consisting only of
    subsets of $V \uplus X \uplus Y \uplus A \uplus B$ and
    satisfying
    $\|\scrS\| = \opt_H(X \reco Y)$.
    In particular,
    any target set in $\scrS$ does not contain any internal vertices of the one-way gadgets of $H$.
\end{claim}
\begin{proof}
Let $\scrS = ( S^{(1)}, \ldots, S^{(\TTT)} )$
be any reconfiguration sequence from $X$ to $Y$ such that
$\|\scrS\| = \opt_H(X \reco Y)$.
Consider then a new sequence
$\scrR = ( R^{(1)}, \ldots, R^{(\TTT)} )$ such that
each $R^{(\ttt)}$ is obtained from $S^{(\ttt)}$
by ``aggregating'' the chosen internal vertices of each one-way gadget into its tail.
Formally,
for a vertex set $S \subseteq V(H)$,
we define $\proj_{\text{\ref{clm:VXYAB}}}(S)$ as
\begin{align}
    \proj_{\text{\ref{clm:VXYAB}}}(S) \triangleq
    \bigl( (V \uplus X \uplus Y \uplus A \uplus B) \cap S \bigr)
    \cup
    \bigl\{
        \text{tail } v \text{ of } D_{v,w} \bigm| \text{some intl.~vert.~of } D_{v,w} \text{ is in } S
    \bigr\}.
\end{align}
Note that $\proj_{\text{\ref{clm:VXYAB}}}(S) \subseteq V \uplus X \uplus Y \uplus A \uplus B$.
Note also that if $S$ and $S'$ differ in at most one vertex,
then $\proj_{\text{\ref{clm:VXYAB}}}(S)$ and $\proj_{\text{\ref{clm:VXYAB}}}(S')$ differ in at most one vertex,
which follows from the \emph{projection} property of $\proj_{\text{\ref{clm:VXYAB}}}$; namely,
$\proj_{\text{\ref{clm:VXYAB}}}(S) = \bigcup_{v \in S} \proj_{\text{\ref{clm:VXYAB}}}(\{v\})$.
Let $R^{(\ttt)} \triangleq \proj_{\text{\ref{clm:VXYAB}}}(S^{(\ttt)})$
for all $1 \leq \ttt \leq \TTT$.
In particular, it holds that
$R^{(1)} = \proj_{\text{\ref{clm:VXYAB}}}(X) = X$ and
$R^{(\TTT)} = \proj_{\text{\ref{clm:VXYAB}}}(Y) = Y$.

Observe that
for any one-way gadget $D_{v,w}$ from $v$ to $w$,
activating its tail $v$ leads to all internal vertices of $D_{v,w}$ being active, implying that
if $S$ is a target set, then so is $\proj_{\text{\ref{clm:VXYAB}}}(S)$; i.e.,
every $R^{(\ttt)}$ must be a target set.
Since $|R^{(\ttt)}| \leq |S^{(\ttt)}|$, and
$R^{(\ttt-1)}$ and $R^{(\ttt)}$ differ in at most one vertex for all $2 \leq \ttt \leq \TTT$,
$\scrR$ is a reconfiguration sequence from $X$ to $Y$ such that
$\|\scrR\| \leq \|\scrS\| = \opt_H(X \reco Y)$, as desired.
\end{proof}

Subsequently, we will discard the vertices of $A \uplus B$, without increasing the size.

\begin{claim}
\label{clm:VXY}
    There exists a reconfiguration sequence $\scrS$ from $X$ to $Y$
    consisting only of subsets of $V \uplus X \uplus Y$ and satisfying
    $\|\scrS\| = \opt_H(X \reco Y)$.
\end{claim}
\begin{proof}
By \cref{clm:VXYAB},
let $\scrS = ( S^{(1)}, \ldots, S^{(\TTT)} )$
be a reconfiguration sequence from $X$ to $Y$
consisting only of subsets of
$V \uplus X \uplus Y \uplus A \uplus B$
such that $\|\scrS\| = \opt_H(X \reco Y)$.
Consider then a new sequence
$\scrR = ( R^{(1)}, \ldots, R^{(\TTT)} )$ such that
each $R^{(\ttt)}$ is obtained from $S^{(\ttt)}$ by ``projecting''
each vertex of $(A \uplus B) \cap S^{(\ttt)}$
onto its destination in $V$ (specified by a one-way gadget leaving from the vertex).
Formally,
for a vertex set $S \subseteq V \uplus X \uplus Y \uplus A \uplus B$,
we define $\proj_{\text{\ref{clm:VXY}}}(S)$ as
\begin{align}
    \proj_{\text{\ref{clm:VXY}}}(S) \triangleq
    \bigl( (V \uplus X \uplus Y) \cap S \bigr)
    \cup
    \bigl\{
        v \in V \bigm| \exists j \text{ s.t. } a_{v,j} \in S \text{ or } b_{v,j} \in S
    \bigr\}.
\end{align}
Equivalently,
$\proj_{\text{\ref{clm:VXY}}}(S)$ is obtained from $S$ as follows:
(1) if $a_{v,j}$ belongs to $S$, then remove $a_{v,j}$ and add $v$;
(2) if $b_{v,j}$ belongs to $S$, then remove $b_{v,j}$ and add $v$.
Note that $\proj_{\text{\ref{clm:VXY}}}(S) \subseteq V \uplus X \uplus Y$.
Note also that if $S$ and $S'$ differ in at most one vertex,
then $\proj_{\text{\ref{clm:VXY}}}(S)$ and $\proj_{\text{\ref{clm:VXY}}}(S')$ differ in at most one vertex,
which follows from the \emph{projection} property of $\proj_{\text{\ref{clm:VXY}}}$; namely,
$\proj_{\text{\ref{clm:VXY}}}(S) = \bigcup_{v \in S} \proj_{\text{\ref{clm:VXY}}}(\{v\})$.
Let
$R^{(\ttt)} \triangleq \proj_{\text{\ref{clm:VXY}}}(S^{(\ttt)})$
for all $1 \leq \ttt \leq \TTT$.
In particular, it holds that
$R^{(1)} = \proj_{\text{\ref{clm:VXY}}}(X) = X$ and
$R^{(\ttt)} = \proj_{\text{\ref{clm:VXY}}}(Y) = Y$.
Since $|R^{(\ttt)}| \leq |S^{(\ttt)}|$,
$R^{(\ttt-1)}$ and $R^{(\ttt)}$ differ in at most one vertex for all $2 \leq \ttt \leq \TTT$, and
$\|\scrR\| \leq \|\scrS\| = \opt_H(X \reco Y)$,
what remains to be seen is that every $R^{(\ttt)}$ is a target set.

For a vertex set $R \subseteq V \uplus X \uplus Y$,
we define $\proj^{-1}_{\text{\ref{clm:VXY}}}(R) \subseteq V \uplus X \uplus Y \uplus A \uplus B$ as
\begin{align}
    \proj^{-1}_{\text{\ref{clm:VXY}}}(R)
    \triangleq R
    \uplus \bigl\{ a_{v,j} \in A \bigm| v \in R \text{ and } 1 \leq j \leq d_v \bigr\}
    \uplus \bigl\{ b_{v,j} \in B \bigm| v \in R \text{ and } 1 \leq j \leq d_v \bigr\}.
\end{align}
Observe that
$\proj^{-1}_{\text{\ref{clm:VXY}}}(R)$ is
the unique inclusion-wise maximal vertex set $S$
such that $\proj_{\text{\ref{clm:VXY}}}(S) = R$.
Thus, it suffices to show that
``if $\proj^{-1}_{\text{\ref{clm:VXY}}}(R)$ is a target set, then so is $R$.''
Suppose (for contraposition) that $R$ is not a target set.
By \cref{obs:target},
$R$ (and thus $\proj^{-1}_{\text{\ref{clm:VXY}}}(R)$) does not contain $X$, $Y$, or any target set for $G$ (entirely);
moreover, $\proj^{-1}_{\text{\ref{clm:VXY}}}(R)$ does not contain $A$ or $B$ (entirely).
We will show that
$\proj^{-1}_{\text{\ref{clm:VXY}}}(R)$ does not activate any
$x_i \in X \setminus \proj^{-1}_{\text{\ref{clm:VXY}}}(R)$
or
$a_{v,j} \in A \setminus \proj^{-1}_{\text{\ref{clm:VXY}}}(R)$.

Consider starting the activation process from $\proj^{-1}_{\text{\ref{clm:VXY}}}(R)$ over
the subgraph of $H$ induced by $V \uplus X \uplus A$ and the one-way gadgets connecting vertices of $V \uplus X \uplus A$.
If the vertices of $\proj^{-1}_{\text{\ref{clm:VXY}}}(R) \cap V$ are initially activated,
then at least one vertex of $V$ would not have been activated
(because $\proj^{-1}_{\text{\ref{clm:VXY}}}(R) \cap V = R \cap V$ is not a target set for $G$).
Since no internal vertices of one-way gadgets are included in
$\proj^{-1}_{\text{\ref{clm:VXY}}}(R)$,
each vertex $x_i$ of $X$ is incident to at most $n-1$ active vertices
(i.e., heads of one-way gadgets from $V$ to $x_i$).
Every vertex $x_i$ of $\proj^{-1}_{\text{\ref{clm:VXY}}}(R) \setminus X$
thus would not get newly activated.
Subsequently,
each vertex $a_{v,j}$ of $A$ is incident to at most $\ell-1$ active vertices
(i.e., heads of one-way gadgets from $X$ to $a_{v,j}$); thus,
every vertex $a_{v,j}$ of $\proj^{-1}_{\text{\ref{clm:VXY}}}(R) \setminus A$ would not get newly activated.
Observe that
for any vertex $v$ of $\proj^{-1}_{\text{\ref{clm:VXY}}}(R) \setminus V$ that has not been activated so far,
we have
$a_{v,j} \notin \proj^{-1}_{\text{\ref{clm:VXY}}}(R)$ for all $1 \leq j \leq d_v$;
as a result, $v$ would not become activated along the one-way gadgets from $a_{v,j}$'s to $v$.
Consequently, $\proj^{-1}_{\text{\ref{clm:VXY}}}(R)$ cannot activate such $v$
during the activation process over
the subgraph of $H$ induced by $V \uplus X \uplus A$ and the one-way gadgets connecting vertices of $V \uplus X \uplus A$.

Similarly,
$\proj^{-1}_{\text{\ref{clm:VXY}}}(R)$ cannot activate any
$y_i \in Y \setminus \proj^{-1}_{\text{\ref{clm:VXY}}}(R)$ or 
$b_{v,j} \in B \setminus \proj^{-1}_{\text{\ref{clm:VXY}}}(R)$
during the activation process over
the subgraph of $H$ induced by $V \uplus Y \uplus B$ and the one-way gadgets connecting vertices of $V \uplus Y \uplus B$,
completing the proof.
\end{proof}

We finally characterize the condition that
a subset of $V \uplus X \uplus Y$ is a target set for $H$.
\begin{claim}
\label{clm:VXY-target}
    For any vertex set $S \subset V \uplus X \uplus Y$,
    $S$ is a target set for $H$ 
    if and only if either
    $X \subseteq S$, $Y \subseteq S$, or $S \cap V$ is a target set for $G$.
\end{claim}
\begin{proof}
Since the ``if'' direction follows from \cref{obs:target},
we prove (the contraposition of) the ``only-if'' direction.
Let
$S \subset V \uplus X \uplus Y$
be any vertex set that satisfies
$X \not\subseteq S$, 
$Y \not\subseteq S$, and
$S \cap V$ is not a target set for $G$.
Let $\calA(S) \subset V$ be
the set of vertices that would have been activated during the activation process starting from $S \cap V$ \emph{over} $G$,
and define $\tilde{S} \triangleq S \cup \calA(S)$.
Since $\tilde{S} \supseteq S$,
it suffices to show that $\tilde{S}$ is not a target set for $H$.

Consider starting the activation process from $\tilde{S}$ over $H$.
We show that no vertices \emph{excepting} those of one-way gadgets can be \emph{newly} activated.
By definition of $\tilde{S}$,
all vertices of $V \setminus \tilde{S}$ would not have become activated by initially activating
$\tilde{S}\cap V = \calA(S)$.
Hence, for each $x_i \in X$,
at most $n-1$ one-way gadgets from $V$ to $x_i$ may be activated,
implying that any vertex $x_i$ of $X \setminus \tilde{S}$ cannot be newly activated.
Plus, for each $a_{v,j} \in A$,
at most $\ell-1$ one-way gadgets from $X$ to $a_{v,j}$ may be activated;
namely, no vertex $a_{v,j}$ of $A \setminus \tilde{S}$ may become activated.
Similarly,
any vertex $y_i$ of $Y \setminus \tilde{S}$ cannot be newly activated, and
no vertex $b_{v,j}$ of $B\setminus \tilde{S}$ may become activated.
Consequently,
(the internal vertices of) all one-way gadgets
connecting from $A$ or $B$ to $V$ remain inactive;
thus, any inactive vertex of $V \setminus \tilde{S}$ would also remain inactive,
completing the proof.
\end{proof}

We are now ready to prove \cref{lem:soundness}.

\begin{proof}[Proof of \cref{lem:soundness}]
Suppose that $\opt_H(X \reco Y) \leq k_s+\ell$.
By \cref{clm:VXY},
there is a reconfiguration sequence from $X$ to $Y$,
denoted by $\scrS = ( S^{(1)}, \ldots, S^{(\TTT)} )$,
consisting only of subsets of $V \uplus X \uplus Y$ such that
$\|\scrS\| = \opt_H(X \reco Y)$.
Note that $\scrS$ \emph{cannot} contain $X \uplus Y$ at any time
because $|X \uplus Y| \geq 2 \ell > k_s+\ell$.
Each target set $S^{(\ttt)}$ in $\scrS$ can thus be classified into the following three types:
\begin{enumerate}[label=\textup{(C\arabic*)}]
    \item $S^{(\ttt)}$ contains $X$ but not $Y$;
    \label{clm:sound:1}
    \item $S^{(\ttt)}$ contains $Y$ but not $X$;
    \label{clm:sound:2}
    \item $S^{(\ttt)}$ contains neither $X$ nor $Y$.
    \label{clm:sound:3}
\end{enumerate}
Obviously,
$S^{(1)} = X$ fits into \ref{clm:sound:1} and
$S^{(\TTT)} = Y$ fits into \ref{clm:sound:2}.
If two target sets fit respectively into \ref{clm:sound:1} and \ref{clm:sound:2},
then their symmetric difference has size at least \emph{two}.
Thus, such a pair of target sets would not be adjacent within $\scrS$; namely,
there must be a target set in $\scrS$ that fits into \ref{clm:sound:3}, denoted by $S^\circ$.

Since $S^\circ \cap V$ must be a target set for $G$ by \cref{clm:VXY-target},
in the subsequence of $\scrS$ from $S^{(1)}$ to $S^\circ$,
there must be some $S^\star \subseteq V \uplus X \uplus Y$ that contains \emph{both} $X$ and any target set for $G$.
Such $S^\star$ can be decomposed into $(S^\star \cap V) \uplus X \uplus Y'$, where
$(S^\star \cap V)$ is a target set for $G$ and $Y' \subseteq Y$,
implying that
\begin{align}
    |S^\star|
    = |(S^\star \cap V) \uplus X \uplus Y'|
    \geq |S^\star \cap V| + |X|
    = |S^\star \cap V| + \ell.
\end{align}
On the other hand, since $\|\scrS\| \leq k_s+\ell$ by assumption, we have
\begin{align}
    |S^\star| \leq k_s+\ell.
\end{align}
Consequently,
$(S^\star \cap V)$ turns out to be a target set for $G$ of size at most $k_s$,
accomplishing the proof.
\end{proof}

\section{Conclusion}
In this paper, we demonstrated that the approximation threshold of \prb{Minmax Target Set Reconfiguration} is $2$.
An immediate open question is
whether \prb{Minmax Target Set Reconfiguration} is \PSPACE-hard to approximate within a factor better than $2$.

\paragraph{Acknowledgments.}
I wish to thank the anonymous referees for their careful reading.

\begin{sloppypar}
\printbibliography 
\end{sloppypar}

\end{document}